\documentclass[11pt]{article}

\usepackage{fullpage}
\usepackage{latexsym}
\usepackage{amssymb,amsfonts,amsmath,amsthm}
\usepackage{url}

\def\01{\{0,1\}}

\newcommand{\eps}{\varepsilon}
\newcommand{\ket}[1]{|#1\rangle}
\newcommand{\bra}[1]{\langle#1|}
\newcommand{\ketbra}[2]{|#1\rangle\langle#2|}

\newcommand{\inp}[2]{\langle{#1}|{#2}\rangle} 

\newcommand{\HM}{\mbox{\rm HM}}
\newcommand{\Maj}{\mbox{\rm Majority}}

\newcommand{\E}{\mathop{\mathbb E}}

\newcommand{\M}{{\cal M}}

\newtheorem{definition}{Definition}
\newtheorem{theorem}{Theorem}

\title{Better Non-Local Games from Hidden Matching}
\author{Harry Buhrman\thanks{CWI and University of Amsterdam, buhrman@cwi.nl. Supported by a Vici grant from NWO.} \and Giannicola Scarpa\thanks{CWI Amsterdam, g.scarpa@cwi.nl. Supported by a Vidi grant from NWO.} \and Ronald de Wolf\thanks{CWI Amsterdam, rdewolf@cwi.nl. Supported by a Vidi grant from NWO.}}

\begin{document}

\maketitle

\begin{abstract}
We construct a non-locality game that can be won with certainty by a quantum strategy using $\log n$ shared EPR-pairs,
while any classical strategy has winning probability at most 
$\frac{1}{2}+O\left(\frac{\log n}{\sqrt{n}}\right)$.
This improves upon a recent result of Junge et al.\ in a number of ways.
\end{abstract}

\section{Introduction}

One of the most striking features of quantum mechanics is the fact that
\emph{entangled} particles can exhibit correlations which cannot be reproduced or
explained by classical physics (i.e., by ``local hidden-variable theories'').
This was first observed by Bell~\cite{bell:epr}
in response to Einstein-Podolsky-Rosen's challenge to the completeness of quantum mechanics~\cite{epr}.
An appropriate experimental realization of such quantum correlations is the strongest proof
we have that nature does not behave in accordance with classical physics.
Many such experiments have in fact already been done starting with~\cite{AspectGR82:Bell}.
All behave in accordance with quantum predictions,
though so far none has closed all conceivable ``loopholes'' that would allow
some (usually very contrived) classical explanation of the observed correlations.

Roughly speaking, the more entanglement the quantum experiment starts with, the
further its exhibited correlations can deviate from what is achievable classically.
In this paper we study this relation quantitatively.
The setup is as follows~\cite{chtw:nonlocal}.
Two spacelike separated parties, called Alice and Bob, receive inputs $x$ and $y$ according
to some fixed and known probability distribution~$\pi$, and are required to produce outputs $a$ and $b$.
There is a predicate $V(ab|xy)$ specifying which output-pairs $a,b$ are considered ``winning'' 
on inputs $x,y$, while the others are ``losing''.
We will usually write $G=G(V,\pi)$ to denote such a game.

Quantum strategies for playing such a game start out with some fixed entangled state,
say with local dimension $n$; a typical example would be $\log n$ shared 
EPR-pairs $\frac{1}{\sqrt{2}}(\ket{00}+\ket{11})$.
For each input $x$, Alice has a set of measurement operators $\{A_x^a\}$
(subject to the usual constraint $\sum_a A_x^a=I$, the $n$-dimensional identity)
and for each $y$, Bob has measurement operators $\{B_y^b\}$ (subject to $\sum_b B_y^b=I$).
They apply the measurement corresponding to $x$ and $y$ to the entangled state $\ket{\psi}$,
producing outputs $a$ and $b$, respectively.
The probability to output $a,b$ is $\bra{\psi}A_x^a\otimes B_y^b\ket{\psi}$.
Note that no communication takes place between Alice and Bob.
Assuming the predicate $V(ab|xy)$ takes value $+1$ on winning outputs and value $-1$
on losing outputs, the \emph{advantage} (i.e., difference between winning and losing probabilities)
of this quantum strategy can be succinctly expressed as:
\[
\omega_q(G,\{A_x^a\},\{B_y^b\},\psi)= \E_{x,y,a,b} [V(ab|xy)] = \sum_{x,y,a,b}\pi(x,y)\bra{\psi}A_x^a\otimes B_y^b\ket{\psi}V(ab|xy).
\]

The quantum value of the game, when restricted to entangled states with local dimension $n$, is defined as
\[
\omega_{q,n}(G)=\max_{\{A_x^a\},\{B_y^a\},\psi\in\mathbb{C}^{n\times n}}\omega_q(G,\{A_x^a\},\{B_y^b\},\psi).
\]
This is a number between $-1$ and 1; a value of~1 indicates that the strategy wins with certainty, 
a value of~$-1$ that it always fails (for instance if $V(ab|xy)=-1$ for all $a,b,x,y$).

This quantum value should be contrasted with the best value that can be obtained by \emph{classical} strategies.
In a classical strategy the shared entangled state is replaced by a shared random variable~$R$, sometimes called the ``hidden variable''.
Its distribution is independent of the inputs and its value $r$ is seen by both Alice and Bob, who can use
this to coordinate their behaviour.
A classical strategy is described by two functions $A:x,r\mapsto a$ and $B:y,r\mapsto b$ which are used
by Alice and Bob, respectively, to determine their output as a function of their input and of the shared random variable.
The value of such a classical strategy for game $G$ is
\[
\omega_c(G,A,B,R)=\sum_{x,y,r}\Pr[R=r]\pi(x,y)V(A(x,r)B(y,r)|xy)
\]
and the classical value of the game is what can be achieved by the best classical strategy:
\[
\omega_c(G)=\max_{A,B,R} \omega_c(G,A,B,R).
\]
It turns out that this value can be achieved by \emph{deterministic} strategies, so we could drop $R$ from this definition.

It is well-known that there are games where $\omega_{q,n}(G)$ is substantially higher than $\omega_c(G)$;
the CHSH game~\cite{chsh} is a famous example of this with $n=2$ (Alice and Bob share one EPR-pair).
In fact, there are games $G$ where the ratio $\omega_{q,n}(G)/\omega_c(G)$ is unbounded when the local dimension $n$ of the entangled state grows.
We are interested in the maximal value this ratio can take as a function of $n$.
Our starting point is a recent paper by Junge et al.~\cite{PerezGarcia09} who studied the same question
using tools from Operator Space theory.
On the one hand, they proved that the ratio cannot be larger than $O(n)$; on the other hand they proved
the existence of a game where the ratio is $\Omega(\sqrt{n}/\log^2 n)$.

Our main result in this paper is a simple game where the local dimension is $n$ and the ratio between
the optimal quantum and classical values is $\Omega(\sqrt{n}/\log n)$:
there is a quantum strategy that achieves the maximal value~1 using $\log n$ EPR-pairs as its entangled state $\ket{\psi}$,
while no classical strategy can have an advantage better than $O(\log n/\sqrt{n})$.
We also give a classical strategy achieving advantage of $\Omega(1/\sqrt{n})$, so our bounds are nearly optimal.

Our game is a variant of the ``Hidden Matching'' problem. This was introduced
in the context of quantum communication complexity by Bar-Yossef et al.~\cite{bjk:q1wayj}, and other variants of it were
subsequently studied in~\cite{gkrw:identificationj,gkkrw:1wayj,gavinsky:interactionvsquantum,gavinsky:interactionvsnonlocality}.
A precise definition will be given below.
The main mathematical tool we use in our analysis is the so-called ``KKL inequality'' from Fourier analysis of Boolean functions
(see~\cite{odonnell:survey,wolf:fouriersurvey} for surveys of this area).
This inequality was used before to analyze another variant of Hidden Matching 
in~\cite{gkkrw:1wayj}, though their analysis is different and more complicated.

Our result has several advantages over the one of Junge et al.~\cite{PerezGarcia09}.
First, our game is simple and explicit, while they only give a non-explicit existence proof
(via a probabilistic argument based on Gaussian matrices).
Clearly, explicitness is necessary for experimental realization.
Second, our quantum-classical separation is slightly stronger, $\Omega(\sqrt{n}/\log n)$ instead of $\Omega(\sqrt{n}/\log^2 n)$.
The stronger the separation, the more resistant it is to noise (for instance, if noise can change numerator and denominator of
a large ratio by some small $\eps$, the ratio won't be significantly affected).
Third, the number of inputs for Alice and Bob is smaller: in their game there are roughly $2^{n\log^2 n}$ possible inputs,
while in our case there are $2^n$ possible inputs $x$ for Alice and roughly $2^{n\log n}$ possible inputs $y$ for Bob
(in fact, the later could easily be reduced to much less than $2^n$). The fewer possible inputs and measurement settings
there are, the easier it should be to experimentally realize a quantum strategy for the game.

The organization of the paper is as follows. While our focus is non-locality, it will actually
be useful to first study the original version of the Hidden Matching problem in the context of protocols
where communication from Alice to Bob is allowed.
In Section~\ref{seccommunication} we prove a tight bound of $1/2+\Theta(c/\sqrt{n})$
on the maximal success probability Alice and Bob can obtain with $c$ bits of classical communication.%
\footnote{This bound includes the main lower bound of~\cite{bjk:q1wayj} as a special case:
namely, the communication $c$ needs to be $\Omega(\sqrt{n})$ bits if we want to have success probability at least 2/3.
The proof of~\cite{bjk:q1wayj} is based on information theory rather than Fourier analysis.}
Section~\ref{secnonlocal} then ports those results from the communication setting to the non-locality setting, establishing the results mentioned above.


\section{The Hidden Matching problem}\label{seccommunication}

\subsection{Problem definition and quantum protocol}

In this section we describe the original Hidden Matching communication problem 
and an efficient quantum protocol for it, both from~\cite{bjk:q1wayj}.

\begin{definition}[Hidden Matching ($\HM$)]
Let $n$ be a power of 2 and $\M_n$ the set of all perfect matchings on the set $[n]=\{1,\ldots,n\}$ (a perfect matching is a partition of $[n]$ into $n/2$ disjoint pairs $(i,j)$, also called ``edges''). Alice is given $x \in\01^n$ and Bob is given $M\in \M_n$, distributed according to the uniform distribution $\mathcal{U}$. We allow 1-way communication from Alice to Bob, and Bob outputs an edge $(i,j) \in M$ and $v\in \01$. They win if $v = x_i \oplus x_j$.
\end{definition}

\begin{theorem}\label{thm:quantumHM}
 There exists a quantum protocol for $\HM$ with $\log n$ qubits of 1-way communication, such that always $v = x_i \oplus x_j$.
\end{theorem}

\begin{proof}
The protocol is the following:
\begin{enumerate}
\item Alice sends Bob the state \( \ket{\psi} = \frac{1}{\sqrt{n}} \sum_{i=1}^n (-1)^{x_i} \ket{i} \).
\item Bob measures $|\psi\rangle$ in the basis $B = \{ \frac{1}{\sqrt{2}} (\ket{i} \pm \ket{j}) \mid (i,j) \in M\}$.
 If the outcome of the measurement is a state $ \frac{1}{\sqrt{2}}(\ket{i} + \ket{j})$ then Bob outputs $(i,j)$ and $v=0$. If the outcome of the measurement is a state $ \frac{1}{\sqrt{2}}(\ket{i} - \ket{j})$, Bob outputs $(i,j)$ and $v=1$. 
\end{enumerate}
For each $(i,j) \in M $ the probability to get outcome $\frac{1}{\sqrt{2}}(\ket{i} + \ket{j})$ is
\[
| \bra{\psi}\frac{1}{\sqrt{2}}(\ket{i} + \ket{j})\rangle |^2 = \left\{\begin{array}{ll}2/n & \mbox{ if }x_i\oplus x_j = 0\\
0 & \mbox{ if }x_i\oplus x_j = 1\end{array}\right.
\]
A similar argument holds for $\frac{1}{\sqrt{2}}(\ket{i} - \ket{j})$. Hence Bob's output is always correct.
\end{proof}

\subsection{Bound on classical communication protocols for $\HM$}

Here we show that classical protocols for Hidden Matching with little communication cannot have a good success probability.

\begin{theorem}\label{thm:upperHM}
Every classical deterministic protocol for $\HM$ with \mbox{$c$ bits} of 1-way communication from Alice to Bob where Bob outputs $(i,j),v$, has
\[
\Pr_{\mathcal{U}} [ v = x_i \oplus x_j ] \leq \frac{1}{2} + O\left(\frac{c}{\sqrt{n}}\right) .
\]
\end{theorem}

The intuition behind the proof is the following.
If the communication $c$ is small, the set $X_m$ of inputs $x$ for which Alice sends message $m$, will typically be large, meaning Bob has little knowledge of most of the bits of $x$.  
By the KKL inequality, this implies that for most of the ${n\choose 2}$ $(i,j)$-pairs, Bob cannot guess the parity $x_i\oplus x_j$ well.  Of course, Bob has some freedom in which $(i,j)$ he outputs, but that freedom is limited to the $n/2$ $(i,j)$-pairs in his matching $M$, and it turns out that on average he won't be able to guess any of those parities well.

\begin{proof}
Fix a classical deterministic protocol. 
For each $m\in\01^c$, let $X_m \subseteq \01^n$ be the set of Alice's inputs for which she sends message $m$.
These sets $X_m$ together partition Alice's input space $\01^n$.
Define $p_m = \frac{|X_m|}{2^n}$.  
Note that $\sum_m p_m = 1$, so $p$ is a probability distribution over the $2^c$ messages $m$.
Define $\eps$ such that $ \Pr_{\mathcal{U}}[v = x_i \oplus x_j] = \frac{1}{2} + \eps $, and 
$\eps_m$ such that $\Pr_{\mathcal{U}}[v = x_i \oplus x_j\mid \mbox{Bob received }m] = \frac{1}{2} + \eps_m $. 
Then $\eps = \sum_m p_m\eps_m$.  

For each $m$ and $(i,j)$ define the probability distribution 
$q_m(i,j) = \Pr_{M \in \M_n}[\mbox{Bob outputs }(i,j) \mid \mbox{Bob received }m]$. 
We have $q_m(i,j) \leq \frac{1}{n-1}$,
because we assume Bob always outputs an edge in $M$ and for fixed $i\neq j$ we have $\Pr_M[(i,j)\in M]=1/(n-1)$
(each $j$ is equally likely to be paired up with $i$).
Also define $\beta^m_{ij} = \E_{x\in X_m}[(-1)^{x_i}\cdot(-1)^{x_j}]$. 
The best Bob can do when guessing $x_i\oplus x_j$ given message $m$, 
is to output the value of $x_i\oplus x_j$ that occurs most often among the $x\in X_m$.
The fraction of $x\in X_m$ where $x_i\oplus x_j=0$ is $1/2 + \beta^m_{ij}/2$, 
hence Bob's optimal success probability when guessing $x_i\oplus x_j$ is $1/2 + |\beta^m_{ij}|/2$. 
This implies, for fixed $m$,
\[
\E_{(i,j)\sim q_m} \left[\frac{1}{2} + \frac{|\beta^m_{ij}|}{2}\right] \geq \Pr_{\stackrel{x \in X_m}{M \in \M_n}} [v = x_i \oplus x_j] = \frac{1}{2} + \eps_m.
\]

As explained in~\cite[section 4.1]{wolf:fouriersurvey}, it follows from the KKL inequality~\cite{kkl:influence} that
\begin{equation}\label{eq:KKL}
\sum_{i,j:i\neq j} (\beta^m_{ij})^2 \leq O\left(\log \frac{1}{p_m}\right)^2.                                                  
\end{equation}
This allows us to upper bound $\eps_m$:
\[
 2\eps_m \leq \E_{(i,j)\sim q_m}[|\beta^m_{ij}|] = \sum_{i,j} q_m(i,j) |\beta^m_{ij}| 
\stackrel{(*)}{\leq} \sqrt{\sum_{i,j}q_m(i,j)^2} \cdot \sqrt{\sum_{i,j}|\beta^m_{ij}|^2} 
\stackrel{(**)}{\leq} \frac{1}{\sqrt{n-1}} \cdot O\left(\log\frac{1}{p_m}\right),
\]
where $(*)$ is Cauchy-Schwarz and $(**)$ follows from 
\(\sum_{i,j} q_m(i,j)^2 \leq \max_{i,j} q_m(i,j) \cdot \sum_{i,j}q_m(i,j) \leq \frac{1}{n-1}\) 
and Eq.~\eqref{eq:KKL}. Now we can bound $\eps$:
\[
 \eps = \sum_m p_m\eps_m \leq \sum_m p_m\frac{O(\log(1/p_m))}{\sqrt{n-1}} = \frac{1}{\sqrt{n-1}} \sum_m p_m  O(\log(1/p_m)) = 
\frac{1}{\sqrt{n-1}} O(H(p)) = O\left(\frac{c}{\sqrt{n}}\right)
\]
where $H$ denotes the entropy function, and $H(p)\leq c$ since the distribution $p$ is on $2^c$ elements.
\end{proof}

\subsection{Classical communication protocol for $\HM$}

Here we design a classical communication protocol that achieves the above upper bound on the success probability.

\begin{theorem}\label{thm:lowerHM}
For every positive integer $c\leq \sqrt{n}$, there exists a classical protocol for $\HM$ with $c$ bits of 1-way communication from Alice to Bob, such that
\[
\Pr_{\mathcal{U}} [ v = x_i \oplus x_j ] = \frac{1}{2} + \Omega\left(\frac{c}{\sqrt{n}}\right) .
\]
\end{theorem}

\begin{proof}
We will first handle the case $c\geq2$, dealing with $c=1$ at the end of the proof. 
Assume for simplicity $\sqrt{n}$ is divisible by $c$. The protocol is as follows:
\begin{enumerate}
\item Alice considers the first $\sqrt{n}$ bits of $x$, and divides them into $c$ consecutive blocks $B_1\ldots B_c$ of $k = \sqrt{n} / c$ bits each. She sends Bob a $c$-bit string $m = m_1\ldots m_c$ where $m_i = \Maj(B_i)$ (the value that occurs more often in $B_i$). Let $B(i)$ denote the index of the block containing~$x_i$. With slight abuse of notation, we will also use $B(i)$ for the corresponding block itself (i.e., a $k$-bit string), and for the set of~$k$ indices of the bits in this block.
\item Let $E$ denote the event that there is an edge $(i,j) \in M$ satisfying $i,j\in [\sqrt{n}], B(i)\neq B(j)$. If $E$ occurs then Bob chooses uniformly at random one of the edges $(i,j)$ satisfying the above condition, and outputs $(i,j)$ and $m_{B(i)}\oplus m_{B(j)}$. If $E$ does not occur then he outputs 
a random $(i,j)\in M$ and a random bit (in which case they win with probability~$1/2$).
\end{enumerate}  
We first show that $\Pr_M[E]\geq 1/10$. 
The probability (over a uniformly random matching $M$) that none of the $i\in[\sqrt{n}]$ is paired up
with a $j\in[\sqrt{n}]$, is at most
\[
\prod_{i=1}^{\sqrt{n}}\frac{n-\sqrt{n}-i+1}{n-i+1}\leq\left(1-\frac{1}{\sqrt{n}}\right)^{\sqrt{n}}\leq 1/e.
\]
On the other hand, if some $i\in[\sqrt{n}]$ is paired up with another $j\in[\sqrt{n}]$,
then since this $j$ will be uniformly distributed over $[\sqrt{n}]\backslash\{i\}$, the probability that
$j$ lands in the same $k$-bit block as $i$ (i.e., that $B(i)=B(j)$) 
is at most $(k-1)/(\sqrt{n}-1)\leq 1/c\leq 1/2$.
Hence $\Pr_M[\neg E]\leq 1/e+1/2<9/10$.

Next we show
\(\Pr_\mathcal{U} [ m_{B(i)}\oplus m_{B(j)} = x_i \oplus x_j \mid E] = \frac{1}{2} + \Omega\left(\frac{c}{\sqrt{n}}\right).\)
Below we condition on $E$ without mentioning this further.
It will be convenient to use $\pm 1$-valued bits instead of $0/1$-valued bits,
because then the parity of two bits corresponds to their product.
Let $X \in \{\pm 1\}^n$, uniformly distributed, be the random variable for Alice's input. 
Let $I,J$ be the random variables for Bob's output edge,  
then $I,J$ are uniformly distributed over distinct blocks $B(I),B(J)$ respectively.
Let $M_{B(I)}\in\{\pm 1\}$ be the majority value of the block $B(I)\in\{\pm 1\}^k$, 
and similarly for $M_{B(J)}$. 
Note that $M_{B(I)}$ has the same sign as the sum of the entries of the block $B(I)$.
 
Since $I$ is uniformly distributed over $B(I)$, the bit $M_{B(I)}$ (which Bob knows) 
has some positive correlation with the bit $X_I$ (which Bob would like to know):
\[
\E_{X,I} [M_{B(I)}X_I] = \E_{X,I} \left[\frac{1}{k} \sum_{\ell \in B(I)} M_{B(I)}X_\ell \right] =
\E_{X,I} \left[\frac{1}{k} \left|\sum_{\ell \in B(I)} X_\ell \right|\right] =
\Omega\left(\frac{1}{\sqrt{k}}\right),
\]
where the last step follows from the binomial distribution: 
the sum of $k$ uniform random coin flips is at least $\sqrt{k}$
away from its expectation with at least constant probability (roughly 5\%).
The same lower bound holds for $\E_{X,J}[M_{B(J)}X_J]$.

Bob uses $M_{B(I)}M_{B(J)}$ to guess the parity $X_IX_J$. 
Because $I$ and $J$ are each uniformly distributed over distinct blocks,
the random variables $M_{B(I)}X_I$ and $M_{B(J)}X_J$ are independent. 
Hence we have the following positive correlation in the guess of the parity:
\[
\E_{X,I,J} [(M_{B(I)}M_{B(J)})(X_IX_J)] = \E_{X,I}[M_{B(I)}X_I]\cdot \E_{X,J}[M_{B(J)}X_J]=  \Omega\left(\frac{1}{\sqrt{k}}\right)\cdot\Omega\left(\frac{1}{\sqrt{k}}\right) = \Omega\left(\frac{1}{k}\right).
\]
This says that, conditioned on the event $E$, 
the players win with probability $\frac{1}{2}+\Omega\left(\frac{c}{\sqrt{n}}\right)$.
In case $E$ does not hold, the winning probability is exactly $1/2$.
Since $\Pr_M[E]\geq 1/10$, the claimed lower bound on the overall winning probability follows.

To handle the case $c=1$, note that if $c=2$ it suffices if Alice sends $m_1\oplus m_2$ instead of $m_1,m_2$. 
This gives a 1-bit protocol with winning probability $1/2+\Omega(1/\sqrt{n})$.
\end{proof}

\section{The Non-Local Hidden Matching problem}\label{secnonlocal}

\subsection{Problem definition and quantum protocol}

We now port the problem and results from Section~\ref{seccommunication} to the non-local setting.

\begin{definition}[Non-Local Hidden Matching ($\HM_{nl}$)]
Let $n$ be a power of 2 and $\M_n$ the set of all perfect matchings on the set $[n]$. Alice is given $x \in \01^n$ and Bob is given $M\in \M_n$, distributed according to the uniform distribution $\mathcal{U}$. Alice's output is a string $a \in \01^{\log n} $ and Bob's output is an edge $(i,j) \in M$ and $b \in \01^{\log n}$. They win the game if the following condition is true 
\begin{equation}\label{eq:HMnlcondition}  
(a \oplus b) (i \oplus j) = x_i \oplus x_j.
\end{equation}
\end{definition}                                                                                 

Next we show $\omega_{q,n}(\HM_{nl}) = 1$, while in the next section we show $\omega_c(\HM_{nl}) = O(\log n / \sqrt{n})$. 
Together this gives our main result: $\omega_{q,n}(\HM_{nl})/\omega_c(\HM_{nl}) = \Omega(\sqrt{n}/\log n)$.

\begin{theorem} \label{thm:quantumHMnl}
There exists a quantum protocol for $\HM_{nl}$ with a shared entangled state of $\log n$ EPR-pairs, such that condition~\eqref{eq:HMnlcondition} is always satisfied.
\end{theorem}

\begin{proof}
The protocol is as follows. Alice and Bob share a state 
\( \ket{\psi}= \frac{1}{\sqrt{n}} \sum_{i\in \01^{\log n}}\ket{i}\ket{i}.\)
\begin{enumerate}
 \item Alice performs a phase-flip to her part of $\ket{\psi}$ according to her input $x$. The state becomes 
  \(\ket{\psi'} = \frac{1}{\sqrt{n}} \sum_{i\in \01^{\log n}} (-1)^{x_i}\ket{i} \ket{i}\).
 \item Bob performs a projective measurement with projectors $P_{ij} = \ketbra{i}{i} + \ketbra{j}{j}$, with $(i,j) \in M$. The state collapses to 
 \(\ket{\psi''} = \frac{1}{\sqrt{2}} [(-1)^{x_i}\ket{i} \ket{i} + (-1)^{x_j}\ket{j} \ket{j}]\) 
for some $(i,j) \in M $.
\item Both players apply Hadamard transforms $H^{\otimes \log{n}}$, so they get
\[
H^{\otimes 2 \log{n}}\ket{\psi''} = 
\frac{1}{\sqrt{2}n} \sum_{a,b \in \01^{\log n}} \left( (-1)^{x_i+(a\cdot i)+(b\cdot i)} + (-1)^{x_j+(a\cdot j)+(b\cdot j)}\right)\ket{a}\ket{b},
\]
\end{enumerate}
where $x\cdot y$ is the bitwise inner product of $x$ and $y$ modulo $2$.
In this last state, only $a,b$ satisfying condition \eqref{eq:HMnlcondition} have nonzero amplitude, hence Alice and Bob win the game with certainty. 
\end{proof}

\subsection{Bound on classical protocols for $\HM_{nl}$}

\begin{theorem} \label{thm:upperHMnl}
 Every classical protocol for $\HM_{nl}$  has
\[
 \Pr_{\mathcal{U}} [ (a \oplus b) (i \oplus j) = x_i \oplus x_j ] \leq \frac{1}{2} + O\left(\frac{\log{n}}{\sqrt{n}}\right).
\]
\end{theorem}

\begin{proof}
A protocol that wins $\HM_{nl}$ with success probability $1/2+\eps$ can be turned into a protocol for $\HM$ with $\log{n}$ bits of communication and the same probability to win:
the players play $\HM_{nl}$, with Alice producing $a$ and Bob producing $i,j,b$;
Alice then sends $a$ to Bob, who outputs $i,j,(a\oplus b)(i\oplus j)$. 
This requires $c=\log{n}$ bits of communication (the length of $a$), so Theorem~\ref{thm:upperHM} gives the bound.
\end{proof}

\subsection{Classical protocol for $\HM_{nl}$}

Here we show that our upper bound of $\frac{1}{2}+O\left(\frac{\log n}{\sqrt{n}}\right)$ on the best success probability
of classical strategies for $\HM_{nl}$ is nearly optimal:

\begin{theorem}\label{thm:lowerHMnl}
There exists a classical protocol for $\HM_{nl}$ such that
\[
\Pr_{\mathcal{U}} [ (a \oplus b) (i \oplus j) = x_i \oplus x_j ] = \frac{1}{2} + \Omega\left(\frac{1}{\sqrt{n}}\right).
\]
\end{theorem}

\begin{proof}
For any positive integer $c\leq \sqrt{n}$, Alice and Bob can ``simulate'' the communication protocol of Theorem~\ref{thm:lowerHM} using shared randomness, as follows. They share a uniformly random string $r\in\01^c$. Alice knows which message $m$ she would have sent to Bob in the communication protocol. If $r=m$ (which happens with probability $1/2^c$) then Alice outputs $a=0^{\log n}$, otherwise she outputs a uniformly random $a\in\01^{\log n}$. Bob treats $r$ as the communication Alice would have sent him in the original protocol, and computes $i,j,v$ accordingly. He outputs $i,j$, and a string $b$ satisfying $b(i\oplus j)=v$. Note that if $m=r$ then $a=0^{\log n}$ and hence $(a\oplus b)(i\oplus j)=v$, in which case Alice and Bob win $\HM_{nl}$ with probability $\frac{1}{2} + \Omega(\frac{c}{\sqrt{n}})$. On the other hand, if $m\neq r$ then $a\oplus b$ is a uniformly random string and $i\oplus j\neq 0^{\log n}$, so $(a\oplus b)(i\oplus j)$ is a uniformly random bit. Hence the overall winning probability is:
\[
\Pr_{\mathcal{U}} [ (a \oplus b) (i \oplus j) = x_i \oplus x_j ] = \frac{1}{2^c} \cdot \left( \frac{1}{2} + \Omega\left(\frac{c}{\sqrt{n}}\right) \right) + \left(1-\frac{1}{2^c}\right) \cdot \frac{1}{2} = \frac{1}{2} + \Omega\left(\frac{c}{2^c\sqrt{n}}\right).                       
\]
Taking $c=1$ gives the claimed result.
\end{proof}

\subsection{An alternative classical protocol based on the Grothendieck inequality}

In this section we describe a classical protocol that works for arbitrary input distributions $\pi$,
instead of just uniform.  This protocol will be based on the famous Grothendieck inequality. 
Let $\HM_{nl}(\pi)$ denote the Non-Local Hidden Matching game with probability distribution $\pi$ on the input.

\begin{theorem}\label{thm:lowerHMnl(pi)}
For every input distribution $\pi$ there exists a classical protocol for $\HM_{nl}(\pi)$ such that
\[
\Pr_{\pi} [ (a \oplus b) (i \oplus j) = x_i \oplus x_j ] = \frac{1}{2} + \Omega\left(\frac{1}{\sqrt{n}\log n}\right).
\]
\end{theorem}

\begin{proof}
Alice and Bob start with a shared uniformly random $i\in[n]$ and $r\in[\log n]$.
There will be a unique $j$ such that $(i,j)\in M$. Bob outputs that $(i,j)$.  
Now we need to explain how they compute $a,b\in\01^{\log n}$ 
such that $(a\oplus b)(i\oplus j)=x_i\oplus x_j$ with advantage $\Omega(1/\sqrt{n}\log n)$.

We start by defining the following systems of unit vectors $\{v_x\},\{v_y\}\subseteq\mathbb{R}^n$, 
and matrix $N\in\mathbb{R}^{n\times n}$.
Here $x$ and $y$ range over $\01^n$ and $\M_n$, respectively, while $i\in[n]$ is as above.
\begin{align*}
v_x & = \frac{1}{\sqrt n} \sum_{k\in \01^{\log n}} (-1)^{x_i\oplus x_k}\ket{k} \\
v_y & = \ket{j}, \mbox{ for } (i,j)\in y \\
N_{xy} & = \pi(x,y) \cdot (-1)^{x_i\oplus x_j}, \mbox{ for } (i,j)\in y.
\end{align*}
We have:
\[
 \sum_{x,y} N_{xy} \inp{v_x}{v_y} = 
 \sum_{x,y} N_{xy} \frac{1}{\sqrt n} \sum_k (-1)^{x_i\oplus x_k}\inp{k}{j} =
 \frac{1}{\sqrt n} \sum_{x,y} N_{xy} (-1)^{x_i\oplus x_j} =  
 \frac{1}{\sqrt n} \sum_{x,y}\pi(x,y) =
 \frac{1}{\sqrt n}.
\]
It follows from the Grothendieck inequality that there is constant%
\footnote{This constant (which is independent of the matrix $N$ and the vectors) is called the Grothendieck constant.
Its exact value is unknown, but~\cite{Krivine79, Davie84, Reed91} show that $1.68 \lesssim K_G \lesssim 1.78$.
The relationship between non-local games and the Grothendieck inequality has been studied extensively in~\cite{BrietBuhrmanToner09}.} 
$K_G$ such that the following holds:
there exist classical strategies $A:x\mapsto\01$ and $B:y\mapsto\01$ such that
\begin{equation}\label{eq:grothendieck}
\sum_{x,y} N_{x,y} (-1)^{A(x)\oplus B(y)} \geq \frac{1}{K_G} \max_{\{v_x\},\{v_y\}}\sum_{x,y} N_{xy} \inp{v_x}{v_y} \geq \frac{1}{K_G\sqrt{n}}.
\end{equation}
This implies that (for every $i$, and for $j$ defined by $(i,j)\in y$), 
we have $\Pr_\pi[A(x)\oplus B(y)=x_i\oplus x_j]\geq \frac{1}{2}+\frac{1}{2K_G\sqrt{n}}$.
It remains to define the output strings $a$ and $b$, which we do as follows:
\begin{itemize}
\item Since Alice knows $i$ and $x$, she knows the bit $A(x)$.
She outputs $a=A(x)e_r$, where $e_r\in \01^{\log n}$ is a string whose only 1-bit sits at position $r$.
\item Since Bob knows $i$ and $y$, he knows the bit $B(y)$. Since $r$ is uniformly random and $i\neq j$,
with probability at least $1/\log n$ the strings $i$ and $j$ differ at position $r$.
If this is the case then Bob outputs $b=B(x)e_r$. This ensures $(a\oplus b)(i\oplus j)=A(x)\oplus B(y)$, 
so then Alice and Bob win with probability at least $1/2+1/(2K_G\sqrt{n})$.
If $i$ and $j$ do not differ at position $r$, then Bob outputs a uniformly random $b\in\01^{\log n}$, 
and they win with probability 1/2. 
\end{itemize}
The overall winning probability is at least $1/2+1/(2K_G\sqrt{n}\log n)$.
\end{proof}

If $\pi$ induces a uniform marginal distribution on $M$, 
then $i$ and $j$ in the above proof differ at position $r$ with probability 1/2 instead of $1/\log n$.
This gives an alternative proof of Theorem~\ref{thm:lowerHMnl}.

\subsection{Quantum/classical ratio as a function of the number of possible outputs}
We can also study the ratio between the quantum and classical values as a function of the number of possible outputs,
rather than the local dimension of the entangled state. Let $\omega_{q}(G)=\sup_n\omega_{q,n}(G)$ be the quantum value 
with unlimited (but finite) entanglement. 
It is known that for a game $G$ with $k$ possible outputs for Alice (i.e., values for $a$) 
and $\ell$ possible outputs for Bob, we have $\omega_q(G)/\omega_c(G)=O(k\ell)$~\cite{dklr:nonsignal,PerezGarcia09}.

As presented above, our game $\HM_{nl}$ has $n$ possible outputs for Alice and roughly $n^3$ for Bob.  
However, it can easily be modified to have only $n$ possible outputs for Bob,
by restricting $\M_n$ to matchings $M=\{(i,j)\}$ where $i\leq n/2$ and $j>n/2$, 
and where there is a bijection between $(i,j)$ and $i\oplus j$.
Now Alice behaves as before, and we just require Bob to output $i\oplus j$ (which he can do in $\log n-1$ bits 
since the most significant bit is always 1), and the bit $w=b(i\oplus j)$. 
The number of possible outputs for Bob is now $2^{\log n-1}\cdot 2=n$,
and the original relation $(a\oplus b)(i\oplus j)=x_i\oplus x_j$ is equivalent to $a(i\oplus j)=x_i\oplus x_j\oplus w$.

For this modified version of $\HM_{nl}$, we also have $\omega_q(G)/\omega_c(G) = \Omega(\sqrt{n}/\log n)$.
First, it is easy to see that the proof of Theorem~\ref{thm:upperHM} works as before,
the only change being that now $q_a(i,j)\leq 2/n$ instead of $\leq 1/(n-1)$,
because we are choosing matchings from a subset of $\M_n$.
Second, the proof of Theorem~\ref{thm:quantumHM} also works when Bob's output is modified as above, 
so we have the claimed bound.

\section{Conclusion and future work}

We presented a simple non-local game where the ratio between the quantum value (with $n$-dimensional entanglement)
and the classical value scales as roughly $\sqrt{n}$.  On the other hand, Junge et al.~\cite{PerezGarcia09}
showed that this ratio is $O(n)$ for all possible games.
It is an interesting open problem to close this quadratic gap: can we find a non-local game where the quantum/classical ratio is $\Omega(n)$ or find a better upper bound on all games?
A second open problem is closing the gap on the ratio between the quantum and classical values as a function of the number of possible outputs.
We have still more than a fourth-power gap with our lower bound of $\Omega(\sqrt{n}/\log n)$.

\subsubsection*{Acknowledgements} 
We thank Jop Bri\"et for useful discussions in the early stages of this research.


\begin{thebibliography}{GKRW09}

\bibitem[AGR82]{AspectGR82:Bell}
A.~Aspect, Ph. Grangier, and G.~Roger.
\newblock Experimental realization of {E}instein-{P}odolsky-{R}osen-{B}ohm
  {G}edankenexperiment: A new violation of {B}ell's inequalities.
\newblock {\em Physical Review Letters}, 49:91, 1982.

\bibitem[BBT09]{BrietBuhrmanToner09}
J.~Bri{\"e}t, H.~Buhrman, and B.~Toner.
\newblock A generalized {G}rothendieck inequality and entanglement in {XOR}
  games.
\newblock Technical report, CWI, 2009.
\newblock arXiv:0901.2009.

\bibitem[Bel65]{bell:epr}
J.~S. Bell.
\newblock On the {E}instein-{P}odolsky-{R}osen paradox.
\newblock {\em Physics}, 1:195--200, 1965.

\bibitem[BJK08]{bjk:q1wayj}
Z.~{Bar-Yossef}, T.~S. Jayram, and I.~Kerenidis.
\newblock Exponential separation of quantum and classical one-way communication
  complexity.
\newblock {\em SIAM Journal on Computing}, 38(1):366--384, 2008.
\newblock Earlier version in STOC'04.

\bibitem[CHSH69]{chsh}
J.~F. Clauser, M.~A. Horne, A.~Shimony, and R.~A. Holt.
\newblock Proposed experiment to test local hidden-variable theories.
\newblock {\em Physical Review Letters}, 23(15):880--884, 1969.

\bibitem[CHTW04]{chtw:nonlocal}
R.~Cleve, P.~H{\o}yer, B.~Toner, and J.~Watrous.
\newblock Consequences and limits of nonlocal strategies.
\newblock In {\em Proceedings of 19th IEEE Conference on Computational
  Complexity}, pages 236--249, 2004.

\bibitem[Dav84]{Davie84}
A.~M. Davie.
\newblock Lower bound for {$K_G$}.
\newblock Unpublished note, 1984.

\bibitem[DKLR08]{dklr:nonsignal}
J.~Degorre, M.~Kaplan, S.~Laplante, and J.~Roland.
\newblock The communication complexity of non-signaling distributions.
\newblock arXiv:0804.4859, 2008.

\bibitem[EPR35]{epr}
A.~Einstein, B.~Podolsky, and N.~Rosen.
\newblock Can quantum-mechanical description of physical reality be considered
  complete?
\newblock {\em Physical Review}, 47:777--780, 1935.

\bibitem[Gav08]{gavinsky:interactionvsquantum}
D.~Gavinsky.
\newblock Classical interaction cannot replace a quantum message.
\newblock In {\em Proceedings of 40th ACM STOC}, pages 95--102, 2008.
\newblock quant-ph/0703215.

\bibitem[Gav09]{gavinsky:interactionvsnonlocality}
D.~Gavinsky.
\newblock Classical interaction cannot replace quantum nonlocality.
\newblock arXiv:0901.0956, 2009.

\bibitem[GKK{\etalchar{+}}08]{gkkrw:1wayj}
D.~Gavinsky, J.~Kempe, I.~Kerenidis, R.~Raz, and R.~{de} Wolf.
\newblock Exponential separation for one-way quantum communication complexity,
  with applications to cryptography.
\newblock {\em SIAM Journal on Computing}, 38(5):1695--1708, 2008.
\newblock Earlier version in STOC'07. quant-ph/0611209.

\bibitem[GKRW09]{gkrw:identificationj}
D.~Gavinsky, J.~Kempe, O.~Regev, and R.~{de} Wolf.
\newblock Bounded-error quantum state identification and exponential
  separations in communication complexity.
\newblock {\em SIAM Journal on Computing}, 39(1):1--24, 2009.
\newblock Special issue on STOC'06. quant-ph/0511013.

\bibitem[JPP{\etalchar{+}}09]{PerezGarcia09}
M.~Junge, C.~Palazuelos, D.~{Perez-Garcia}, I.~Villanueva, and M.~Wolf.
\newblock Unbounded violations of bipartite {B}ell inequalities via {O}perator
  {S}pace theory.
\newblock arXiv:0910.4228, 2009.

\bibitem[KKL88]{kkl:influence}
J.~Kahn, G.~Kalai, and N.~Linial.
\newblock The influence of variables on {B}oolean functions.
\newblock In {\em Proceedings of 29th IEEE FOCS}, pages 68--80, 1988.

\bibitem[Kri79]{Krivine79}
J.~L. Krivine.
\newblock Constantes de {G}rothendieck et fonctions de type positif sur les
  sph{\`e}res.
\newblock {\em Advances in Mathematics, 31:16–30}, pages 16--30, 1979.

\bibitem[{O'D}08]{odonnell:survey}
R.~{O'Donnell}.
\newblock Some topics in analysis of boolean functions.
\newblock Technical report, ECCC Report TR08--055, 2008.
\newblock Paper for an invited talk at STOC'08.

\bibitem[Ree91]{Reed91}
J.~A. Reeds.
\newblock A new lower bound on the real {G}rothendieck constant.
\newblock Available at \url{http://www.dtc.umn.edu/~reedsj}, 1991.

\bibitem[Wol08]{wolf:fouriersurvey}
R.~{de} Wolf.
\newblock A brief introduction to {F}ourier analysis on the {B}oolean cube.
\newblock {\em Theory of Computing}, 6, 2008.
\newblock ToC Library, Graduate Surveys.

\end{thebibliography}
\bibliographystyle{alpha}
\newcommand{\etalchar}[1]{$^{#1}$}

\end{document}